\documentclass[3p]{elsarticle}
\usepackage{amssymb,amsmath,amsfonts,amsthm,dsfont}
\usepackage[normalem]{ulem}
\usepackage{bbm}
\usepackage{bm}
\usepackage{physics}
\usepackage{graphicx,wasysym}
\usepackage{xcolor}
\usepackage{float}
\usepackage{mathtools}
\newcommand{\sgn}[1]{\text{sgn}\left\lbrack#1\right\rbrack}
\newcommand{\testx}{\tilde{\mathbf{x}}}
\newcommand{\x}{\mathbf{x}}
\newcommand{\y}{\mathbf{y}}
\newcommand{\V}{\mathcal{V}}
\newcommand{\Zal}{\mathcal{Z}_{al}}

\newtheorem{definition}{Definition}
\newtheorem{proposition}{Proposition}
\newtheorem{lemma}{Lemma}

\newcommand{\RR}{\mathbbm{R}}
\newcommand{\NN}{\mathbbm{N}}

\newcommand{\CC}{\mathbbm{C}}
\newcommand{\FF}{\mathbbm{F}}

\biboptions{sort&compress}

\makeatletter
\def\ps@pprintTitle{%
  \let\@oddhead\@empty
  \let\@evenhead\@empty
  \def\@oddfoot{\reset@font\hfil\thepage\hfil}
  \let\@evenfoot\@oddfoot
}
\makeatother
\begin{document}

\title{The theory of the quantum kernel-based binary classifier}

\author[1,2]{Daniel K. Park\fnref{fn1}}
\ead{dkp.quantum@gmail.com}

\author[3]{Carsten Blank\fnref{fn1}} \ead{blank@data-cybernetics.com}

\author[4,5]{Francesco Petruccione} \ead{petruccione@ukzn.ac.za}

\address[1]{School of Electrical Engineering, KAIST, Daejeon, 34141, Republic of Korea}
\address[2]{ITRC of Quantum Computing for AI, KAIST, Daejeon, 34141, Republic of Korea}
\address[3]{Data Cybernetics, 86899 Landsberg, Germany}
\address[4]{Quantum Research Group, School of Chemistry and Physics, University of KwaZulu-Natal, Durban, KwaZulu-Natal, 4001, South Africa}
\address[5]{National Institute for Theoretical Physics (NITheP), KwaZulu-Natal, 4001, South Africa}

\fntext[fn1]{These authors contributed equally to this work.}

\begin{abstract}
Binary classification is a fundamental problem in machine learning. Recent development of quantum similarity-based binary classifiers and kernel method that exploit quantum interference and feature quantum Hilbert space opened up tremendous opportunities for quantum-enhanced machine learning. To lay the fundamental ground for its further advancement, this work extends the general theory of quantum kernel-based classifiers. Existing quantum kernel-based classifiers are compared and the connection among them is analyzed. Focusing on the squared overlap between quantum states as a similarity measure, the essential and minimal ingredients for the quantum binary classification are examined. The classifier is also extended concerning various aspects, such as data type, measurement, and ensemble learning. The validity of the Hilbert-Schmidt inner product, which becomes the squared overlap for pure states, as a positive definite and symmetric kernel is explicitly shown, thereby connecting the quantum binary classifier and kernel methods.
\end{abstract}


\maketitle
\def\one{{\mathchoice {\rm 1\mskip-4mu l} {\rm 1\mskip-4mu l} {\rm \mskip-4.5mu l} {\rm 1\mskip-5mu l}}}
\section{Introduction}

Quantum properties such as its representation in exponentially large Hilbert space and quantum interference provide compelling opportunities for quantum-enhancement in various domains of computing, communication and sensing. Quantum algorithms and quantum error correcting codes promise quantum advantages in certain computational tasks over existing classical methods. Nevertheless, demonstration of the quantum advantage in practical and industrial applications remains on the agenda.

A discipline for which quantum techniques are expected to be beneficial is machine learning~\cite{wittek,doi:10.1080/00107514.2014.964942,QML-Biamonte,SupervisedQML,Dunjko_2018,doi:10.1098/rspa.2017.0551,Preskill2018quantumcomputingin}. Machine learning has garnered much attention recently due to the rapid accumulation of data and the demand for efficient ways to extract useful information. The emergence of quantum machine learning, which bridges quantum computing and machine learning, appears to be natural for several instinctive reasons. For example, quantum computing can provide rich data representation in exponentially large quantum state space and efficiently perform matrix operations in this high-dimensional vector space as machine learning tasks often require~\cite{PhysRevLett.103.150502_HHL_qBLAS}. Also, since quantum physics can generate patterns with the properties, such as superposition and entanglement, that cannot be described with classical physics, quantum theory may also enable pattern recognition beyond classical capabilities~\cite{Romero_2017}. Moreover, the theory of quantum state discrimination~\cite{Helstrom1969,Barnett:09,Bae_2015} which has been studied both intensively and extensively may find its applications in machine learning problems, such as classification and decision making.

Pattern recognition is a fundamental problem in machine learning with broad applications. In pattern analysis, the kernel method has been regarded as an eminent tool for identifying non-linear relationships in data. It uses a kernel function, i.e. a similarity measure of data, that is associated with reproducing kernel Hilbert space in which the kernel can be evaluated by taking an inner product~\cite{Scholkopf:2000:KTD:3008751.3008793,hofmann2008}. Quantum computers are expected to improve existing classical kernel-based machine learning methods for their ability to efficiently access and manipulate data in large quantum feature spaces, which is classically intractable. Indeed, one of the earliest quantum machine learning algorithms developed to provide exponential speedup in certain cases is the quantum support vector machine~\cite{PhysRevLett.113.130503_qSQVM}, a supervised machine learning algorithm based on the kernel method. More recent work in Ref.~\cite{PhysRevLett.122.040504} established theoretical foundations of the kernel method for quantum machine learning. 

Binary classification is an example of pattern recognition. The goal of this task can be described as, given a complex-valued labelled data set $\mathcal{D} = \left\{ (\x_1, y_1), \ldots, (\x_M, y_M) \right\} \subset \mathbb{C}^N\times\{0,1\}$, finding the most likely class label of an unseen data point $\testx \in \mathbb{C}^N$. Here, the data is assumed to be complex-valued, rather than real-valued as in usual machine learning tasks, to utilize the quantum Hilbert space to the full extent. Supervised learning methods that exploit quantum Hilbert space to represent the feature space are proposed and verified in Refs.~\cite{Havlicek2019,2019arXiv190610467S}, advocating potential quantum speedups in machine learning with these approaches. Quantum binary classifiers based on evaluating the sum of many kernel functions with only a constant number of runs by exploiting quantum interference are introduced in Refs.~\cite{QML_Maria_Francesco,2019arXiv190902611B}. Ref.~\cite{2019arXiv190902611B} describes the quantum circuit construction for tailoring the weights and exponents of kernel functions defined by the quantum state fidelity of pure states.

This paper aims to present an up-to-date snapshot of the fundamental research in quantum kernel-based classification and the understanding of its current benefits and limitations. For such purposes, we review and compare several known quantum similarity (or kernel)-based binary classifiers. Then we extend the general theory of quantum kernel-based classifiers in particular for the swap-test classifier introduced in Ref.~\cite{2019arXiv190902611B} that uses the quantum fidelity between two pure states as the similarity measure. We report relevant and important remarks that are not discussed in previous works to lay the ground and foster further research. The remainder of the paper is organized as follows. Section~\ref{sec:review} provides a review of the quantum kernel-based classification algorithms and discusses their relations. Section~\ref{sec:disect} examines the swap-test classifier to extract essential ingredients for the algorithm and to suggest a minimal quantum circuit implementation. Section~\ref{sec:general} discusses the generalization of the swap-test classifier to density matrix formalism, ensemble learning and single-shot measurement. In Sec.~\ref{sec:KernelMethods}, the properties of the Hilbert-Schmidt inner product, which becomes the squared overlap or the quantum state fidelity for pure states, in relation to the kernel theory are investigated and its validity as a kernel function is discussed. Conclusion and suggestions for future work are provided in Sec.~\ref{sec:conclusion}.
 
\section{Quantum kernel-based Classification}
\label{sec:review}

In this section, we review and discuss three known quantum classification algorithms, namely the quantum support vector machine (qSVM)~\cite{PhysRevLett.113.130503_qSQVM}, the Hadamard classifier (HC)~\cite{QML_Maria_Francesco} and the swap-test classifier (STC)~\cite{2019arXiv190902611B}. 
There is a common and underlying agreement of the encoding of data that they adhere to: a vector $\x = (x_1, \ldots, x_N)^T \in \CC^N$ is encoded in a quantum state
\begin{equation}
\label{eq:amplitude_encoding}
    \ket{\x} := \frac{1}{\|\x\|} \sum_{i=1}^{N} x_i \ket{i},
\end{equation}
which is also known as the \textit{amplitude encoding} of data~\cite{QML_Maria_Francesco}. 

Throughout the paper, we take $N = 2^n$ for some integer $n$ without loss of generality, i.e., the data set is encoded in $n$ qubits. We reserve $\|\cdot\|$ for the $l_2$-norm of a vector, $\braket{\cdot}{\cdot}$ for the inner product of two pure quantum states, and $\langle \cdot, \cdot \rangle$ for the inner product of an arbitrary complex-valued Hilbert space, which is linear in the second argument and semi-linear in the first.

\subsection{Quantum Support Vector Machine}
As the support vector machine is a commonly used algorithm in machine learning we want to recall that $f(\x) = \langle \mathbf{w}, \x \rangle + b$ is the linear regression function for two vectors $\mathbf{w}$ and $\x$ in some Hilbert space $\mathcal{H}$ and a bias $b\in\mathbb{R}$. The sign of the regression function is then the binary classifier $c(x) = \sgn{f(x)}$. Given a data set $\mathcal{D}$, the question arises if a weight vector $\mathbf{w}$ can be found such that the two classes are separated optimally. The theory is standard in machine learning (see, e.g., Ref.~\cite{hofmann2008, shawetaylor2004kernel, Scholkopf:2000:KTD:3008751.3008793, scholkopf2001learning, hofmann2008_kernel_methods_ml}). There are two equivalent minimization formulations, the primal and dual. In the latter case --- the much more frequently used --- one maximizes a Lagrangian 
\begin{equation}
    \mathcal{L}(\mathbf{a}, \mathbf{w}) = \frac{1}{2} \| \mathbf{w} \|^2 - \sum_i a_i \left( y_i \left( \langle \mathbf{w}, \x_i \rangle + b \right) - 1 \right)
\end{equation}
with respect to the Lagrange multipliers $a_i \geq 0$ and applies the Karush-Kuhn-Tucker (KKT) conditions, which are the necessary and sufficient conditions for a maximum in this convex optimization problem. The KKT conditions then state that the weight vector can be represented as $\mathbf{w} = \sum_i a_i l_i \x_i$ with $l_i = (-1)^{y_i}$ and $\mathbf{a}^\top\mathbf{l} = 0$. It follows that the Lagrangian can be written as
\begin{equation}
    \label{eq:lagrangian}
    \mathcal{L}(\mathbf{a}) = \sum_i a_i - \frac{1}{2} \sum_{i,j} a_j l_j \langle \x_i, \x_j \rangle a_i l_i.
\end{equation}
The solutions to the maximization problem denoted by $a^*_i$ are now part of the regression function, which becomes
\begin{equation}
\label{eq:svm_regression_fnc}
    f(\x) = \sum_j a^*_j l_j \langle \x_j, \x \rangle + b = \sum_j (-1)^{y_j} a^*_j \langle \x_j, \x \rangle + b.
\end{equation}
The bias $b$ can be recovered by finding one index $s\in\{1, \ldots, M\}$ for which $a^*_s>0$ (which are called support vectors) and then applying the formula $b = y_s - \langle \mathbf{w}, \x_s \rangle$. We will from now on drop the $*$ from $a^*_i$ for the optimal Lagrange multiplier. We note that there are different formulations of the above, but the one that has been applied in Ref.~\cite{PhysRevLett.113.130503_qSQVM} uses modified multipliers $\alpha_i = a_i l_i$.

The qSVM algorithm introduced in Ref.~\cite{PhysRevLett.113.130503_qSQVM} makes a least-squares approximation of the problem~\cite{Suykens1999} and employs the density matrix exponentiation~\cite{qPCA} and the quantum matrix inversion algorithm~\cite{PhysRevLett.103.150502} to produce a quantum state $\ket{b,\bm{\alpha}}=\left(b\ket{0}+\sum_{m=1}^M\alpha_m\ket{m}\right)/\sqrt{C}$, where $\sqrt{C}$ is the normalization constant. Given $\ket{\x_m}$ for $m=1,\ldots ,M$ and $\ket{\tilde\x}$ representing the training data and the test datum, respectively, and $\ket{b,\bm{\alpha}}$, the qSVM classifier is defined by the oracle state
\begin{equation}
    \ket{\Phi} = \frac{1}{\sqrt{2}} \left( \ket{0} \ket{\tilde{u}} + \ket{1}\ket{\tilde{x}} \right), \label{eq:qsvm_initial_state}
\end{equation}
where the first register is an ancilla qubit and
\begin{align}
    \ket{\tilde{u}} &:= N_{\tilde{u}}^{-\frac{1}{2}} \left( b \ket{0}\ket{0} + \sum_{m=1}^M \alpha_m \|\x_m\| \ket{m} \ket{\x_m} \right), \label{eq:qsvm_training_data}\\
    \ket{\tilde{x}} &:= N_{\tilde{x}}^{-\frac{1}{2}} \left( \ket{0}\ket{0} + \sum_{m=1}^M \|\testx\| \ket{m} \ket{\testx} \right) \label{eq:qsvm_test_data}
\end{align}
with $N_{\tilde{u}} = |b|^2 + \sum_m |\alpha_m|^2 \|\x_m\|^2$ and $N_{\tilde{x}} = M \|\tilde\x\|^2 + 1$. Applying a Hadamard gate on the first (ancilla) qubit of $\ket{\Phi}$ and measuring the ancilla in the state $1$ yields the probability $\Pr(a=1) = \left( 1 - \Re \braket{\tilde{u}}{\tilde{x}} \right)/2$, where
\begin{equation*}
    \braket{\tilde{u}}{\tilde{x}} = (N_{\tilde{u}}N_{\tilde{x}})^{-\frac{1}{2}} \left( b + \sum_{m=1}^M \alpha_m \|\x_m\| \|\testx\| \braket{\x_m}{\testx} \right).
\end{equation*}
Following the methods of Ref.~\cite{2019arXiv190902611B}, one actually finds that measuring the expectation value of $\sigma_z$ on the ancilla gives
\begin{equation}
    \bra{\Phi}H_a \sigma_z^{(a)} H_a \ket{\Phi} = \Re \braket{\tilde{u}}{\tilde{x}},
\end{equation}
and the classification is done by the sign function of the expectation value, i.e., $c(\testx) = \sgn{\Re \braket{\tilde{u}}{\tilde{x}}}$.

The proof-of-principle of the qSVM was demonstrated with a four-qubit nuclear spin quantum processor for classifying handwritten characters in Ref.~\cite{PhysRevLett.114.140504}.

\subsection{Hadamard classifier}
The Hadamard classifier~\cite{QML_Maria_Francesco} is an interesting alternative to the qSVM classifier that bypasses costly subroutines such as the density matrix exponentiation and the quantum matrix inversion while it adds a separate qubit for the class label and uses a different measurement scheme. The main interest of this work was to design a quantum classifier that can be realized by a minimum quantum circuit. This classifier demonstrated that a simple quantum circuit consisting of a Hadamard gate and a post-selection measurement scheme can realize binary similarity-based classification. Only during recent research it was shown that the post-selection measurement can be replaced by an expectation value measurement of a two-qubit observable and that training and test data can be bestowed with weights so that it actually closely resembles the qSVM classifier~\cite{2019arXiv190902611B}. The proof-of-principle of the HC was demonstrated with a publicly available five-qubit quantum computer provided by the IBM Quantum Experience for a simplified supervised pattern recognition task based on the famous
Iris flower~\cite{QML_Maria_Francesco}.

We now put this work into perspective and explain the general procedure. An initial state is created with weights $a_m \geq 0$
\begin{equation}
\label{eq:hc_initial_state}
    \ket{\Phi_h} = \frac{1}{\sqrt{2}} \left( \ket{0} \ket{\tilde{u}_h} + \ket{1}\ket{\tilde{x}_h} \right)
\end{equation}
where the first register is an ancilla qubit and
\begin{align}
    \ket{\tilde{u}_h} &:= N_{\tilde{u}_h}^{-\frac{1}{2}} \left( \sum_{m=1}^M \sqrt{a_m} \|\x_m\| \ket{m} \ket{\x_m} \ket{y_m} \right), \label{eq:hc_training_data}\\
    \ket{\tilde{x}_h} &:= N_{\tilde{x}_h}^{-\frac{1}{2}} \left( \sum_{m=1}^M \sqrt{a_m} \|\testx\| \ket{m} \ket{\testx} \ket{y_m} \right) \label{eq:hc_test_data}
\end{align}
and the normalizing constants are given by $N_{\tilde{u}_h} = \sum_m a_m \|\x_m\|^2$ and $N_{\tilde{x}_h} = \|\testx\| \sum_m a_m$. The last register is the label qubit that indicates the class label $y_m\in\lbrace 0,1\rbrace$ of an $m$th training data. In the original works of the HC and the STC, the data vectors are normalized to unit length, i.e., $\|\x_m\|= \|\testx\|=1$. However, we explicitly write these terms in this section to make the comparison to the qSVM more clearly. Moreover, when the HC was first introduced, the weights were all equal, i.e., $a_m=1/M\;\forall\; m=1,\ldots,M$. Here we allow for non-uniform weights so as to make a direct connection to the qSVM. As in the qSVM classifier, the HC classifier also applies a Hadamard gate on the ancilla qubit but then performs a two-qubit measurement on the ancilla and the label qubit. For brevity, we use the notation $\Zal = \sigma_z^{(a)}\sigma_z^{(l)}$ throughout this paper, where the superscript $a$ ($l$) indicates that the operator is acting on the ancilla (label) qubit . The resulting expectation value is given by
\begin{align}
     \bra{\Phi_h}H_a \Zal H_a \ket{\Phi_h} &= \Re \bra{\tilde{u}_h} \sigma_z^{(l)} \ket{\tilde{x}_h} \nonumber \\
        &= (N_{\tilde{u}_h} N_{\tilde{x}_h})^{-\frac{1}{2}} \sum_m (-1)^{y_m} a_m \|\testx\| \|\x_m\| \Re \braket{\x_m}{\testx}.
\end{align}

The qSVM classifier has a bias $b\in\RR$ in its description. The bias can be recovered in the Hadamard classifier by encoding $b$ in a state orthogonal to index state vectors as
\begin{align}
    \ket{\tilde{u}'_h} &= N_{\tilde{u}'_h}^{-\frac{1}{2}} \left( \sqrt{b}\ket{0} \ket{0} \ket{y_b} + \sum_{m=1}^M \sqrt{a_m} \|\x_m\| \ket{m} \ket{\x_m} \ket{y_m} \right), \\
    \ket{\tilde{x}'_h} &= N_{\tilde{x}'_h}^{-\frac{1}{2}} \left( \sqrt{b}\ket{0} \ket{0} \ket{y_b} + \sum_{m=1}^M \sqrt{a_m} \|\testx\| \ket{m} \ket{\testx} \ket{y_m} \right)
\end{align}
with constants $N_{\tilde{u}'_h} = |b| + N_{\tilde{u}_h},\; N_{\tilde{x}'_h} = |b| + N_{\tilde{x}_h}$ and $y_b = (1 - \text{sgn}(b))/2$. Then we would indeed find for $\alpha_m = (-1)^{y_m} a_m$
\begin{equation}
     \langle \Zal \rangle = ( N_{\tilde{u}'_h} N_{\tilde{x}'_h} )^{-\frac{1}{2}} \left( b + \sum_m \alpha_m \|\testx\|\, \|\x_m\| \Re \braket{\x_m}{\testx} \right).
\end{equation}

Interesting to note is the striking similarity of both classifiers. However, the qSVM classifier needs one less qubit. What seems at first as redundant opened up a new application and resulted in the discovery of the swap-test classifier~\cite{2019arXiv190902611B}. As we show in Sec.~\ref{sec:disect}, a separate label register is necessary to enable the use of the squared state overlap (i.e. $|\braket{\x}{\testx}|^2$) as a similarity measure to distinguish class labels. The squared state overlap gives a solution to a major drawback of both previous classifiers. The qSVM and the HC are based on estimating only the real part of the inner product even though both works in Refs.~\cite{PhysRevLett.113.130503_qSQVM, QML_Maria_Francesco} mention the power of kernel methods with quantum feature maps. If however, quantum feature maps are to be utilized in their full power, all of the complex space must be accessible to the classifier.

\subsection{Swap-test Classifier}
The swap-test classifier (STC) starts with a product state consisting of a test state $\ket{\testx}$ and a training state (data, label and index) $\ket{\tilde u_s} = N_{\tilde u_s}^{-\frac{1}{2}} \sum_{m=1}^M \sqrt{a_m} \|\x_m\| \ket{\x_m} \ket{y_m}\ket{m}$ with $N_{\tilde u_s} = \sum_m a_m \|\x_m\|^2$. Then a swap-test is done on the test and training data qubits with the help of an ancilla qubit. A two-qubit measurement is applied to the ancilla and label qubits. The setup is summarized in the circuit shown in Fig.~\ref{fig:1}.
\begin{figure}[t]
    \centering
    \includegraphics[width=0.5\columnwidth]{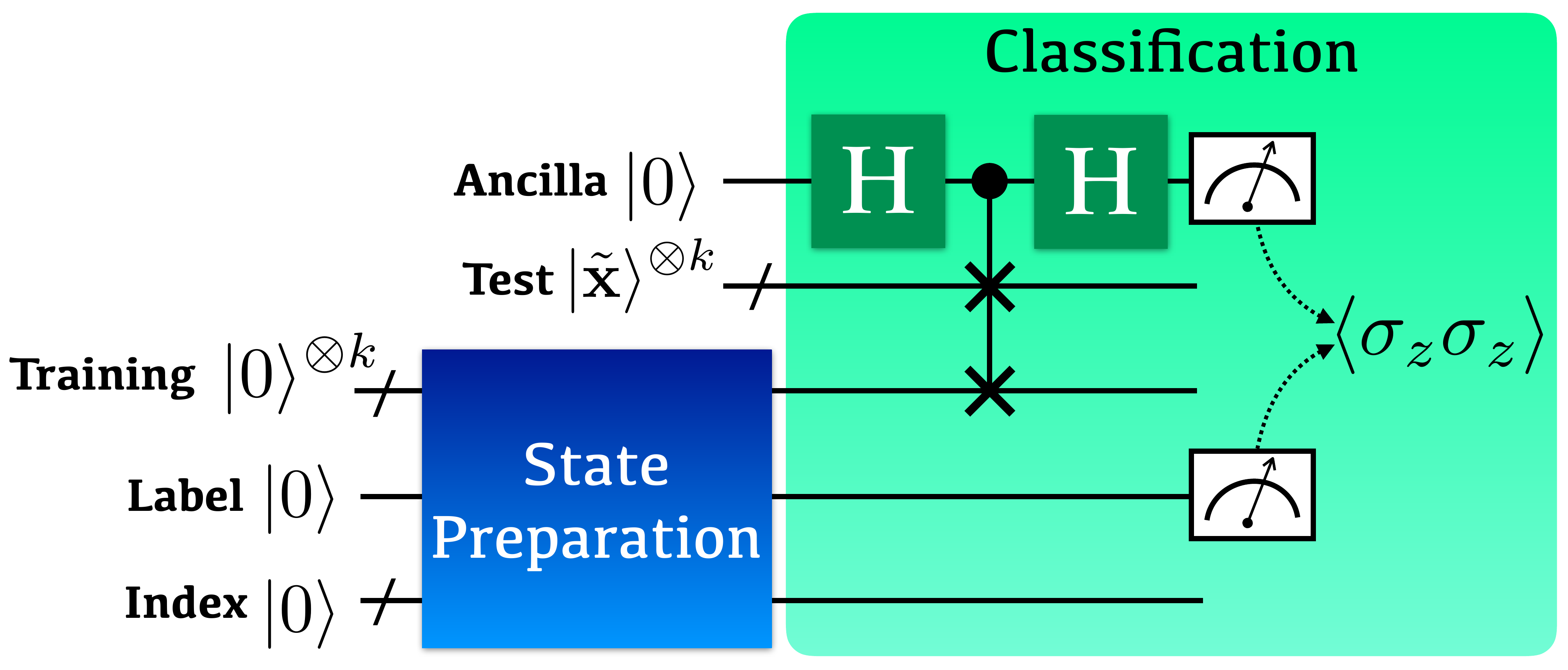}
    \caption{The quantum circuit describing the swap-test classifier.}
    \label{fig:1}
\end{figure}
Then the expectation value is
\begin{equation}
    \langle \Zal \rangle = N_{\tilde u_s}^{-\frac{1}{2}} \sum_{m=1}^M (-1)^{y_m} a_m \|\x_m\|^{2k} \left| \braket{\x_m}{\testx}\right|^{2k}.
\end{equation}
The power $k > 0$ can be achieved by providing $k$ identical copies of the test and training data.

As in the case of the HC we also want to recover the SVM formulation. For this we need to encode the bias $b \in \RR$. We modify the training state $\ket{\tilde u_s}$ to be
\begin{equation}
    \ket{\tilde u'_s} = N_{\tilde u'_s}^{-\frac{1}{2}} \left( \sqrt{b}  \ket{\testx}^{\otimes k}\ket{y_b}\ket{0} + \sum_{m=1}^M \sqrt{a_m} \|\x_m\|^{k}  \ket{\x_m}^{\otimes k} \ket{y_m}\ket{m} \right).
\end{equation}
with $N_{\tilde u'_s} = |b| + \sum_m a_m \|\x_m\|^{2k}$. Then the expectation value of the STC measurement imitates the regression function of Eq.~(\ref{eq:svm_regression_fnc}):
\begin{equation}
    \langle \Zal \rangle = N_{\tilde u'_s}^{-\frac{1}{2}} \left( b + \sum_{m=1}^M (-1)^{y_m} a_m \|\x_m\|^{2k} \left| \braket{\x_m}{\testx}\right|^{2k} \right).
\end{equation}

The STC can be recognized as an algorithm that outputs an expectation value of the Helstrom operator, $p_0\rho_0-p_1\rho_1$, by defining $\rho_i=\sum_{m|y_m=i}(a_m/p_i)\ketbra{\x_m}{\x_m}^{\otimes k}$, where $\sum_{m|y_m=i} a_m/p_i=1$ and $p_0+p_1=1$~\cite{2019arXiv190902611B}. Although the Helstrom operator is the basis from which the optimal binary state discrimination strategy is derived~\cite{Helstrom1969}, the STC does not exactly perform the Helstrom measurement since the Helstrom measurement requires the projection on to positive and negative eigenspaces of the Helstrom operator while the STC acts like a black box that merely outputs the expectation value. Nevertheless, the connection between the Helstrom operator and the STC stimulates further research on applying the mature field of the quantum state discrimination to quantum classification problems. Indeed, the Helstrom measurement has been adapted in the development of a classical quantum-inspired binary classification algorithm~\cite{cagliari.hqc}.


\section{Dissecting the swap-test classifier}
\label{sec:disect}
\subsection{Classification without index}
\label{sec:without_ind}
Recall that the swap-test classifier introduced in Ref.~\cite{2019arXiv190902611B} begins with an input quantum state\footnote{In Ref.~\cite{2019arXiv190902611B}, the weight is denoted as $w_m$. Here, we use $a_m$ to avoid confusion with the convention for the weight vector $\mathbf{w}$ in SVM.}
\begin{equation}
\label{eq:swap_test_in}
\sum_{m=1}^M\sqrt{a_m}\ket{0}\ket{\tilde{\x}}^{\otimes k}\ket{\x_m}^{\otimes k}\ket{y_m}\ket{m},
\end{equation}
where the first qubit is an ancilla for the swap-test, $k$ copies of test and training data are encoded in the quantum state $\ket{\tilde{\x}}$ and $\ket{\x_m}$ by amplitude encoding, $m$ is an index for training data and $y_m\in\lbrace 0,1 \rbrace$ represents the class label.
The next step of the classifier is to apply
\begin{equation}
\mathcal{V}=H_a\cdot\prod_{i=1}^k\text{c-}\texttt{swap}(t_i,d_i\vert a=1)\cdot H_a,
\end{equation}
where $H_a$ represents a Hadamard gate applied to the ancilla qubit and $\text{c-}\texttt{swap}(t_i,d_i\vert a)$ represents a controlled-swap gate that exchanges an $i$th copy of test ($t_i$) and training ($d_i$) data if the ancilla qubit state is $a$.
This results in the final state
\begin{equation}
\label{eq:swap_test_out}
\sum_{m=1}^M\frac{\sqrt{a_m}}{2}(\ket{0}\ket{\psi_{k+}}+\ket{1}\ket{\psi_{k-}})\ket{y_m}\ket{m},
\end{equation}
where $\ket{\psi_{k\pm}}=\ket{\tilde{\x}}^{\otimes k}\ket{\x_m}^{\otimes k}\pm\ket{\x_m}^{\otimes k}\ket{\tilde{\x}}^{\otimes k}$.

The index qudit used in the original work of the HC~\cite{QML_Maria_Francesco} and the STC~\cite{2019arXiv190902611B} can be useful for loading classical information to quantum states using state preparation routines such as quantum random access memory (QRAM)~\cite{PhysRevLett.100.160501,ffqram}, and also is the key element of the STC when all test, training, and label qubits are provided as a product state in separate registers. However, in principle, the index qudit is not strictly necessary in both classifiers. This can be seen by noting that the index qudit is not involved in any operation once the quantum state in Eq.~(\ref{eq:swap_test_in}) is prepared.

Now, let's consider the following scenario. There exists a quantum oracle that outputs $k$ copies of a training datum and its class label as a product state with a probability $a_m$. Then the output of such an oracle, together with an ancilla qubit initialized in $\ket{0}$ and the test data, can be expressed as
\begin{equation}
    \rho_\text{init}^k\coloneqq\ketbra{0}{0}\otimes\ketbra{\tilde{\x}}{\tilde{\x}}^{\otimes k}\otimes\sum_{m=1}^{M} \left( a_m\ketbra{\x_m}{\x_m}^{\otimes k}\otimes\ketbra{y_m}{y_m} \right).
\end{equation}
Applying $\mathcal{V}$ to the above state produces
\begin{align}
    \mathcal{V}\rho_\text{init}^k\mathcal{V}^{\dagger}= \frac{1}{4}&\sum_{i\in{\lbrace 0,1\rbrace}}\ketbra{i}{i}\otimes\Bigg{[}\sum_{m=1}^{M}a_m\ketbra{\tilde{\x}}{\tilde{\x}}^{\otimes k}\otimes\ketbra{\x_m}{\x_m}^{\otimes k}\otimes\ketbra{y_m}{y_m}\nonumber \\
  & +(-1)^i\sum_{m=1}^{M}a_m\ketbra{\tilde{\x}}{\x_m}^{\otimes k}\otimes\ketbra{\x_m}{\tilde{\x}}^{\otimes k}\otimes\ketbra{y_m}{y_m}\nonumber\\ &+(-1)^i\sum_{m=1}^{M}a_m\ketbra{\x_m}{\tilde{\x}}^{\otimes k}\otimes\ketbra{\tilde{\x}}{\x_m}^{\otimes k}\otimes\ketbra{y_m}{y_m}\nonumber\\
   &+ \sum_{m=1}^{M}a_m\ketbra{\x_m}{\x_m}^{\otimes k}\otimes\ketbra{\tilde{\x}}{\tilde{\x}}^{\otimes k}\otimes\ketbra{y_m}{y_m}\Bigg{]}+\ldots .
\end{align}
In the above equation, only the diagonal part of the ancilla qubit subspace is shown explicitly since only these terms contribute when the expectation measurement is performed on the ancilla qubit. The expectation measurement of $\mathcal{Z}_{al}=\sigma_z^{(a)}\sigma_z^{(l)}$ on the ancilla qubit and the label qubit results in
\begin{align}
    \langle\mathcal{Z}_{al}\rangle &= \frac{1}{2} \sum_{m=1}^Ma_m\Big{[}\Tr\left(\ketbra{\tilde{\x}}{\x_m}^{\otimes k}\otimes\ketbra{\x_m}{\tilde{\x}}^{\otimes k}\right)+\Tr\left(\ketbra{\x_m}{\tilde{\x}}^{\otimes k}\otimes\ketbra{\tilde{\x}}{\x_m}^{\otimes k}\right)\Big{]} \Tr\left(\sigma_z\ketbra{y_m}{y_m}\right)\nonumber\\
    &=\sum_{m=1}^Ma_m(-1)^{y_m}|\braket{\tilde\x}{\x_m}|^{2k},
\end{align}
which is identical to the outcome of the original swap-test classifier circuit introduced in Ref.~\cite{2019arXiv190902611B}.

\subsection{Classification without ancilla}
The unitary operator $\mathcal{V}$ followed by measuring the expectation of the observable $\Zal$ is equivalent to measuring the expectation value of an observable. This is easy to see since $\Tr\left(\Zal\V\rho_{\text{init}}^k\V^{\dagger}\right)=\Tr\left(\V^{\dagger}\Zal\V\rho_{\text{init}}^k\right)$. The effective observable can be expressed as
\begin{equation}
\V^{\dagger}\Zal\V=\sigma_z\otimes\left(\frac{1}{2}\sum_{i=0}^3\sigma_i\otimes\sigma_i\right)^{\!\!\otimes nk}\otimes\sigma_z,
\end{equation}
where $n$ is the number of qubits needed for representing the data, $\sigma_i$ corresponds to an $X$, $Y$, and $Z$ Pauli operator for $i=1$, 2, and 3, respectively, and $\sigma_0$ is the identity matrix. Since the initial state of the ancilla qubit is the $+1$ eigenstate of the first term of the effective observable, the measurement on the ancilla qubit simply produces a factor $+1$ and hence can be neglected. Therefore, the swap-test classifier can be generalized as measuring an expectation of an observable
\begin{equation}
\label{eq:O}
    \mathcal{O}\coloneqq \left(\frac{1}{2}\sum_{i=0}^3\sigma_i\otimes\sigma_i\right)^{\!\!\otimes nk}\otimes\sigma_z
\end{equation}
on $\ketbra{\tilde{\x}}{\tilde{\x}}^{\otimes k}\otimes\sum_{m=1}^{M}a_m\ketbra{\x_m}{\x_m}^{\otimes k}\otimes\ketbra{y_m}{y_m}$.
The condensed (minimal) version of the STC is depicted in Fig.~\ref{fig:2}.
\begin{figure}[t]
    \centering
    \includegraphics[width=0.45\columnwidth]{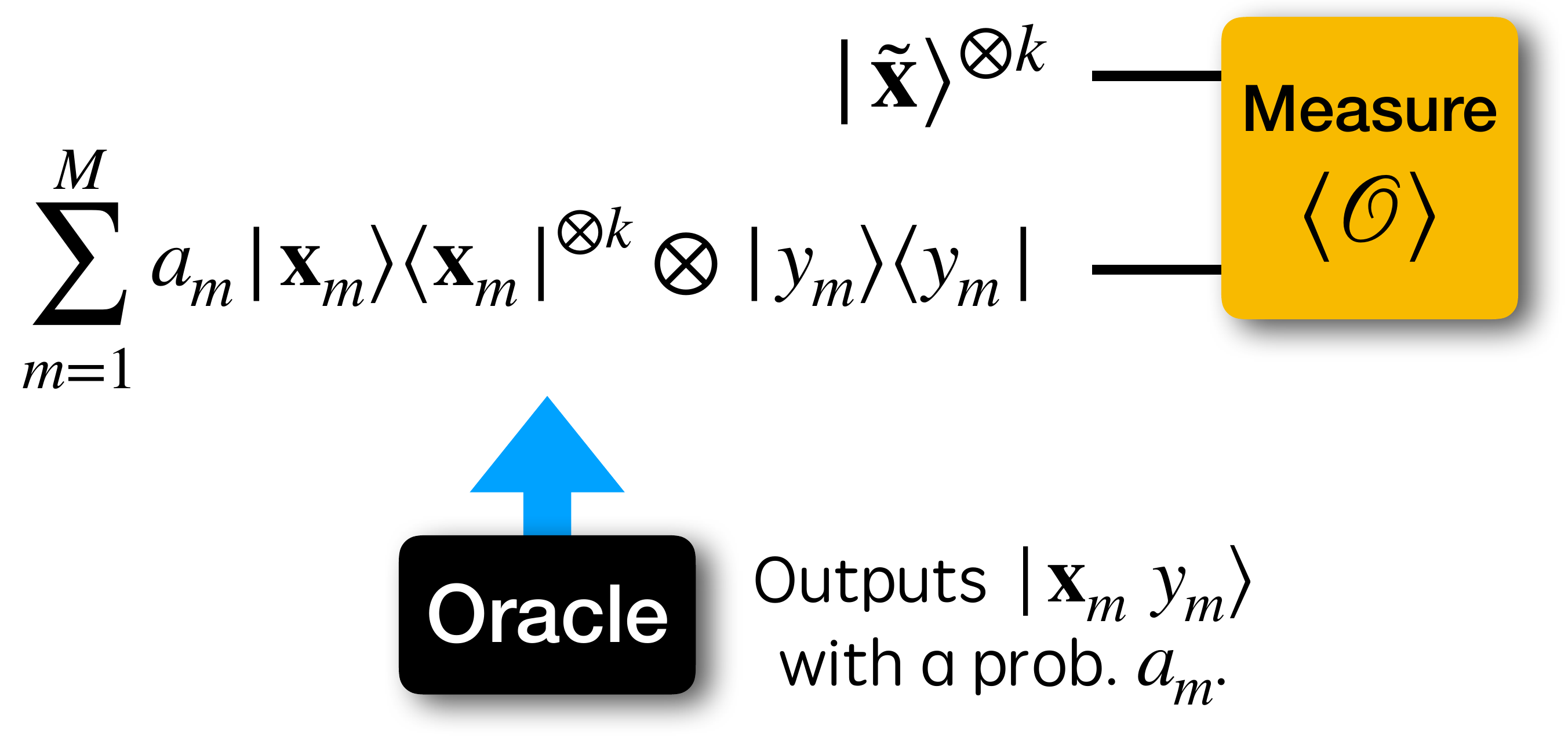}
    \caption{Schematic of the minimal quantum state fidelity-based classifier.}
    \label{fig:2}
\end{figure}
This version is particularly useful when $nk=1$, since in this case the classification can be performed by measuring only four Pauli observables without requiring an ancilla qubit and the quantum gates. In practice, the swap-test with an ancilla becomes useful when $nk$ is large since the number of Pauli operators to be measured increases exponentially with $n$ and $k$. However, if the observable of Eq.~(\ref{eq:O}) can be measured directly in a given experimental setting, the condensed version can be advantageous.

\section{Generalizations of the Swap-test Classifier}
\label{sec:general}
\subsection{Density Matrix Encoding}
\label{sec:mixed}
The STC algorithm proposed in Ref.~\cite{2019arXiv190902611B} assumes that all data are encoded in a pure state and is based on the state vector formalism. In general, test and training data can be encoded in density matrices. Here we generalize the swap-test classifier using the density matrix formalism. This can be useful for describing classification tasks involving arbitrary quantum states, noisy input data, etc. The density matrices describing test data and an $m$th training data are given as 
\begin{equation}
\label{eq:data_mixed}
\tilde\rho=\sum_{i}\tilde p_i\ketbra{\tilde\x_i}{\tilde\x_i}\text{ and }\rho_m=\sum_j p_{j,m}\ketbra{\x_{j,m}}{\x_{j,m}},
\end{equation}
respectively, where $\sum_i \tilde{p}_{i} = 1,\; \tilde{p}_{i}\geq 0$ and $\sum_j p_{j,m} = 1,\; p_{j,m}\geq 0$ for all $m=1, \ldots, M$. We use the similar setting as in Sec.~\ref{sec:without_ind} to assume that an oracle outputs $k$ copies of $\rho_m\otimes \ketbra{y_m}{y_m}$ with the probability of $a_m$. Then the initial state can be written as a density matrix
\begin{equation}
\label{eq:mixed_init}
    \ketbra{0}{0}\otimes\sum_{m=1}^M \left( a_m\left(\tilde\rho\otimes\rho_m\right)^{\otimes k}\otimes\ketbra{y_m}{y_m} \right).
\end{equation}
Note that the training and test data have been reordered in the above equation for notational convenience, but this does not affect the main results. After applying $\V$, the above state becomes
\begin{align}
\frac{1}{4}&\sum_{i\in\lbrace0,1\rbrace}\ketbra{i}{i}\otimes\bigg{[}\sum_{m=1}^Ma_m\left(\tilde\rho\otimes\rho_m\right)^{\otimes k}\otimes\ketbra{y_m}{y_m}\nonumber\\
&+(-1)^i\sum_{m=1}^Ma_m \left(\sum_{i,j}\tilde p_ip_{j,m}\ketbra{\tilde{\x}_i}{\x_{j,m}}\otimes \ketbra{\x_{j,m}}{\tilde{\x}_i}\right)^{\otimes k}\otimes\ketbra{y_m}{y_m}\nonumber\\
&+(-1)^i\sum_{m=1}^Ma_m\left(\sum_{i,j}\tilde p_ip_{j,m}\ketbra{\x_{j,m}}{\tilde{\x}_i}\otimes\ketbra{\tilde{\x}_i}{\x_{j,m}}\right)^{\otimes k}\otimes\ketbra{y_m}{y_m}\nonumber\\
&+\sum_{m=1}^Ma_m\left(\rho_m\otimes \tilde\rho\right)^{\otimes k}\otimes\ketbra{y_m}{y_m}\bigg{]}+\ldots .
\end{align}
Again, only the diagonal part of the ancilla qubit subspace is shown explicitly since only these terms contribute when an expectation measurement is performed on the ancilla qubit. The expectation measurement of $\mathcal{Z}_{al}=\sigma_z^{(a)}\sigma_z^{(l)}$ results in
\begin{equation}
\label{eq:mixed}
    \langle\mathcal{Z}_{al}\rangle =\sum_{m=1}^Ma_m(-1)^{y_m}\left(\sum_{i,j}\tilde p_ip_{j,m}|\braket{\tilde{\x}_i}{\x_{j,m}}|^2\right)^k=\sum_{m=1}^Ma_m(-1)^{y_m}\Tr(\tilde\rho\rho_m)^k.
\end{equation}
Since density matrices are Hermitian, non-negative and have trace $1$, they are of \textit{trace class} and the trace of the product of two density matrices is the Hilbert-Schmidt inner product~\cite{reed2012methods, conway2000course}, which is denoted by $\langle \tilde\rho, \rho_m\rangle_{HS}$. The above result clarifies that the STC is in fact based on the Hilbert-Schmidt inner product. The Hilbert-Schmidt inner product and the quantum state fidelity coincide for pure states. 

For general mixed states, the Hilbert-Schmidt inner product can misguide the classification in some cases. For example, if training and test data qubits are $\rho_m=\one/2$ and $\tilde\rho=\epsilon\ketbra{0}+(1-\epsilon)\ketbra{1}$, respectively, the fidelity of two quantum states is $\Tr(\sqrt{\sqrt{\one/2}\tilde\rho\sqrt{\one/2}})^2=1/2+\sqrt{\epsilon(1-\epsilon)}$, resulting in 1/2 for a pure state (i.e., $\epsilon=1$ or 0) and 1 for the maximally mixed state. On the other hand, the Hilbert-Schmidt inner product between any test density matrix and the training data is 1/2 and fails to provide a proper measure of similarity between the two. An important and interesting open question that stems from the above observation is how to design the quantum state fidelity-based binary classifier for general mixed states.

Unlike the STC, the qSVM or the HC cannot easily be generalized for mixed state input data. As shown in Eqs.~(\ref{eq:qsvm_initial_state}) and~(\ref{eq:hc_initial_state}), the qSVM and the HC encode the training and test data in the same qubit register, but separately in orthogonal subspaces of an ancilla qubit. Then the underlying principle of the classification is the interference of these subspaces. However, such encoding does not naturally extend to the data in mixed states. For example, as a naive extension of the construction shown in Eqs.~(\ref{eq:qsvm_initial_state}) and~(\ref{eq:hc_initial_state}) to mixed states, imagine that training and test data are encoded in the same register as $p_0\ketbra{0}{0}\otimes\rho_{training}+p_1\ketbra{1}{1}\otimes\rho_{test}$, $p_0+p_1=1$. This state describes a classical mixture of having either training or test data, and obviously, they cannot interfere nor be correlated quantum mechanically. To exploit quantum interference between two density matrices for mixed states, they should be encoded in separate registers and be correlated. The STC fulfils this requirement.

\subsection{Ensemble classifier}
Ensemble methods effectively combine multiple learning agents to improve the overall machine learning performance beyond the ability of any constituent learning agent~\cite{10.1007/3-540-45014-9_1}. Figure~\ref{fig:3} shows a pictorial representation of the ensemble method. In this section, we present the framework to formulate an ensemble classifier based on the STC.
\begin{figure}[t]
    \centering
    \includegraphics[width=0.4\columnwidth]{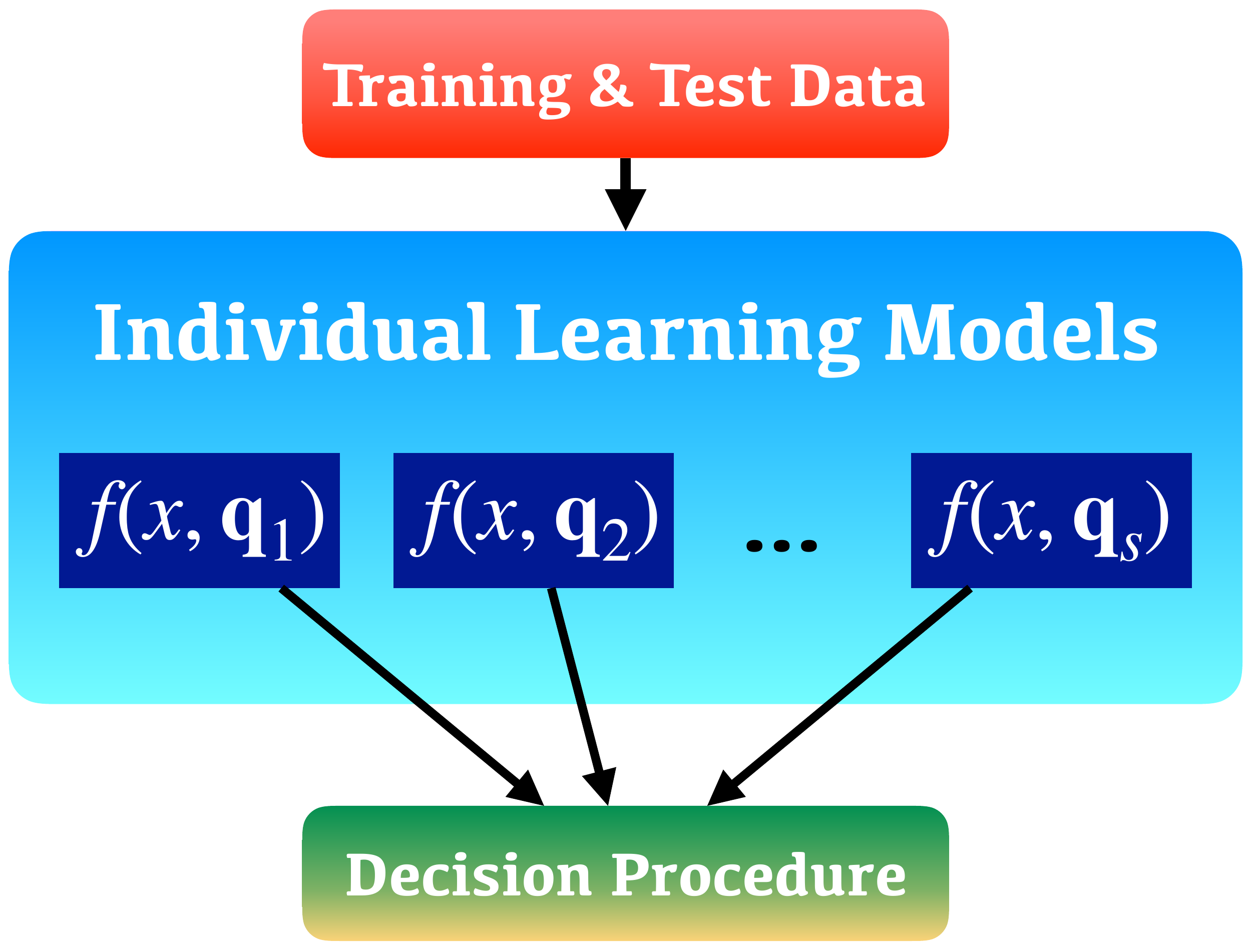}
    \caption{Pictorial description of an ensemble method in supervised learning. An ensemble learning method takes a set of individual learning methods, denoted by $f(x,\mathbf{q}_i)$ where $\mathbf{q}_i$ represents the parameters that determine the model, and employs a decision procedure to derive a new classifier that performs better than any $f$ by itself.}
    \label{fig:3}
\end{figure}
\subsubsection{Ensemble-average over weights}
The result shown in Eq.~(\ref{eq:mixed}) can be generalized even further to construct an ensemble average over $S$ different models of classifiers that are parametrized by the weights $\mathbf{a}$. If an initial state can be written as a density matrix (cf. Eq.~(\ref{eq:mixed_init}))
\begin{equation}
\label{eq:ens_w}
   \ketbra{0}{0}\otimes \sum_{s=1}^S q_s\left(\sum_{m=1}^M a_{m,s}\left(\tilde\rho\otimes\rho_m\right)^{\otimes k}\otimes\ketbra{y_m}{y_m} \right),
\end{equation}
where $\sum_s q_s = 1,\; q_s\ge 0$, then by linearity, the expectation measurement of $\mathcal{Z}_{al}$ results in
\begin{equation}
\label{eq:ens1}
    \langle\mathcal{Z}_{al}\rangle=\sum_{s=1}^Sq_s\sum_{m=1}^Ma_{m,s}(-1)^{y_m}\langle\tilde\rho,\rho_m\rangle_{HS}^k=\sum_{s=1}^Sq_s f(\tilde\rho,\mathbf{a}_s),
\end{equation}
where $f(\tilde\rho,\mathbf{a}_s)$ is a regression function for the test data $\tilde\rho$ with a model parameter, the weights, determined by $\mathbf{a}_s$.

\subsubsection{Ensemble-average over exponents}
An ensemble-average can be performed over classifiers with different exponent as follows. If an initial state can be prepared as (cf. Eqs.~(\ref{eq:mixed_init}) and~(\ref{eq:ens_w}))
\begin{equation}
\label{eq:ens_exp}
    \ketbra{0}{0}\otimes\sum_{s=1}^Sq_s\left(\sum_{m=1}^Ma_{m,s}\left(\tilde\rho\otimes\rho_m\right)^{\otimes k_s}\otimes\left(\one/2\right)^{\otimes (S-k_s)}\otimes\ketbra{y_m}{y_m}\right),
\end{equation}
then by linearity, the expectation measurement of $\mathcal{Z}_{al}$ results in
\begin{equation}
\label{eq:ens2}
    \langle\mathcal{Z}_{al}\rangle=\sum_{s=1}^Sq_s\sum_{m=1}^Ma_{m,s}(-1)^{y_m}\langle\tilde\rho,\rho_m\rangle_{HS}^{k_s}=\sum_{s=1}^Sq_s f(\tilde\rho,\mathbf{a}_s,k_s),
\end{equation}
where again $\sum_s q_s = 1,\; q_s\ge 0$ and $f(\tilde\rho,\mathbf{a}_s,k_s)$ is a regression function for the test data $\tilde\rho$ with two model parameters, the weights and the exponent, determined by $\mathbf{a}_s$ and $k_s$, respectively. Note that the ancillary density matrix $\left(\one/2\right)^{\otimes (S-k_s)}$ can be replaced with any $2^{S-k_s}$ by $2^{S-k_s}$ density matrix.

\subsection{Classification via projective measurement}
The eigenstates of the swap operator are $\ket{00}$, $\ket{11}$, $(\ket{01}+\ket{10})/\sqrt{2}$ and $(\ket{01}-\ket{10})/\sqrt{2}$, with eigenvalues 1, 1, 1, and -1, respectively. Thus when $nk=1$, the spectral decomposition of $\mathcal{O}$ can be written as
\begin{equation}
\label{eq:proj_m}
    \mathcal{O}=\sum_{i=1}^8\lambda_i\ketbra{\lambda_i}{\lambda_i},
\end{equation}
where $\lambda_i=1$ for $i\le 4$ and $\lambda_i=-1$ for $i > 4$, and
\begin{align}
\label{eq:eigen8}
    \ket{\lambda_1}=\ket{000},\; \ket{\lambda_2}=\ket{110},&\;
    \ket{\lambda_3}=\frac{\ket{01}+\ket{10}}{\sqrt{2}}\otimes\ket{0},\; \ket{\lambda_4}=\frac{\ket{01}-\ket{10}}{\sqrt{2}}\otimes\ket{1},\nonumber \\
    \ket{\lambda_5}=\ket{001},\; \ket{\lambda_6}=\ket{111},&\;
    \ket{\lambda_7}=\frac{\ket{01}+\ket{10}}{\sqrt{2}}\otimes\ket{1},\; \ket{\lambda_8}=\frac{\ket{01}-\ket{10}}{\sqrt{2}}\otimes\ket{0}.
\end{align}
For $nk>1$, an eigenstate can be obtained by taking a tensor product between $\ket{\lambda_i}$ and an eigenstate of the swap operator. For simplicity, we use $nk=1$ to illustrate the underlying idea. With the above spectral decomposition, the projective measurement of $\mathcal{O}$ can be described as follows. Given the density matrix $\rho=\tilde\rho\otimes\sum_ma_m\rho_m\otimes\ketbra{y_m}$, the probability to obtain the measurement outcome of $\lambda\in\lbrace+1,-1\rbrace$ is
\begin{align}
\label{eq:Pr_a1}
    \Pr[\lambda]&=\Tr\left(\rho\sum_{i|\lambda_i=\lambda}\ketbra{\lambda_i}\right)\nonumber\\
    &=\frac{1}{2}\left\lbrace\Tr\lbrack\rho_{\lambda}\left(\one\otimes\one+\mathcal{S}\right)\rbrack+\Tr\lbrack(\Tr_l(\rho)-\rho_\lambda\right)\left(\one\otimes\one-\mathcal{S}\right)\rbrack\rbrace,
\end{align}
where $\Tr_l$ represents the partial trace over the label qubit, $\rho_\lambda = \Tr_l(\rho \ketbra{l_\lambda}{l_\lambda})=\tilde\rho \otimes \sum_{m|y_m=l_\lambda}a_m\rho_m$ with $l_\lambda\coloneqq (1-\lambda)/2$, and $\mathcal{S}$ is the swap operator. We can further simplify the above result to obtain
\begin{align}
\label{eq:Pr_a2}
    \Pr[\lambda]&=\frac{1}{2}\left(1+\sum_{m|y_m=l_\lambda}a_m\Tr(\tilde\rho\rho_m)-\sum_{m|y_m\neq l_\lambda}a_m\Tr(\tilde\rho\rho_m)\right) \nonumber\\
    &=\frac{1}{2} \left( 1 + \lambda \sum_{m} (-1)^{y_m} a_m
    \langle\tilde\rho,\rho_m\rangle_{HS}\right).
\end{align}
Therefore, the STC can be described by a projective measurement that outputs $\lambda=+1$ or $-1$ with the probability $\Pr[\lambda]$ shown in Eq.~(\ref{eq:Pr_a2}). One can easily verify that the measurement outcome averages to Eq.~(\ref{eq:mixed}):
\begin{equation*}
    \langle \mathcal{Z}_{al} \rangle = -1 \Pr[\lambda=-1] + 1 \Pr[\lambda=1] = \sum_{m} (-1)^{y_m} a_m
    \langle\tilde\rho,\rho_m\rangle_{HS}.
\end{equation*}
This calculation gives a rigorous understanding of the single-shot behavior of the STC.

Given the above projective measurement, the error probability of misclassification can also be calculated. For this, we assume that the test data can be expressed as $\tilde\rho=p_0\tilde\rho_0+p_1\tilde\rho_1$, meaning that with a probability of $p_i$, the correct class of the test data is $i$. Then the misclassification probability is
\begin{align}
    \text{Pr}_{error} &= p_0\Tr\Bigg{(}\rho_0\otimes \sum_{m}^M a_m\rho_m\otimes\ketbra{y_m}\!\sum_{i|\lambda_i=-1}\!\ketbra{\lambda_i}\Bigg{)}\nonumber \\
    &\; +p_1\Tr\Bigg{(}\rho_1\otimes \sum_{m}^M a_m\rho_m\otimes\ketbra{y_m}\!\sum_{i|\lambda_i=+1}\!\ketbra{\lambda_i}\Bigg{)}\nonumber \\
    &=\frac{1}{2}\Bigg{(}1-\sum_{m=1}^Ma_m\Tr(\tilde\rho\rho_m)\Bigg{)}+p_0\!\sum_{m|y_m=1}\!a_m\Tr(\tilde\rho_0\rho_m)+p_1\!\sum_{m|y_m=0}\!a_m\Tr(\tilde\rho_1\rho_m)\nonumber\\
    \label{eq:proj_error}
    &=\frac{1}{2}\Bigg{(}1+\sum_{m|y_m=1}a_m\langle p_0\tilde\rho_0-p_1\tilde\rho_1,\rho_m\rangle_{HS}+\sum_{m|y_m=0}a_m\langle p_1\tilde\rho_1-p_0\tilde\rho_0,\rho_m\rangle_{HS}\Bigg{)}.
\end{align}
The error probability equation is not only useful for benchmarking the performance of the single-shot STC, but also striking as the Helstrom operator again appears in it. This result promotes future work on understanding the fundamental link between the STC and quantum state discrimination. 

\section{Kernel Methods}
\label{sec:KernelMethods}
All classifiers discussed in Sec.~\ref{sec:review} can in principle be combined with an encoding of the data that can formally be expressed as a map $\Phi : \mathcal{X} \rightarrow \mathcal{H}$ of a data domain $\mathcal{X}$ into a quantum Hilbert space $\mathcal{H}$. This map is commonly called a \textit{feature map} in the domain of machine learning and can have tremendous effects in solving non-linear classification problems~\cite{Havlicek2019}. 

In Sec.~\ref{sec:mixed}, Eq.~(\ref{eq:mixed}) shows that the underlying similarity measure is the Hilbert-Schmidt inner product, which becomes the squared overlap function when the data is given in pure states as in Ref.~\cite{2019arXiv190902611B}. To fully enjoy the kernel theory, the similarity measures must satisfy certain conditions: they must be \textit{symmetric} and \textit{positive definite} kernels. Then Mercer's theorem and the Representer theorem can be applied and the connection to existing theory~\cite{scholkopf2001generalized, shawetaylor2004kernel, scholkopf2001learning} can be made.

\begin{definition}\label{def:psd}
Given a function $\kappa: \mathcal{X} \times \mathcal{X} \rightarrow \CC$ from a set $\mathcal{X}$, the function is called positive semi-definite (PSD) when for a finite set of elements $S=\{\x_1, \ldots, \x_M\} \subset \mathcal{X}$ and all $c_n \in \mathbbm{C}$, $n=1, \ldots, M$, the Gram matrix is positive semi-definite, i.e.,
\begin{equation}
    \sum_{n,m=1}^M \overline{c}_n c_m \kappa(\x_n,\x_m) \geq 0.
\end{equation}
\end{definition}
When the kernel is real-valued, symmetry is also required, i.e., $\kappa(\x, \y) = \kappa(\y, \x)$. If a kernel is PSD (often referred to as positive definite kernel) and symmetric, one can construct a feature map $\Phi: \mathcal{X} \rightarrow \mathcal{H}_R$ with a Hilbert space $\mathcal{H}_R$, the reproducing kernel Hilbert space (RKHS), in which the kernel $\kappa$  can be represented as an inner product~\cite{hofmann2008_kernel_methods_ml}. The construction is summarized here. The RKHS is a subspace of all complex functionals on $\mathcal{X}$, i.e., $\mathcal{H}_R \subseteq \CC^\mathcal{X}$:
\begin{equation}
    \mathcal{H}_R = \left\{ f : f(\cdot) = \sum_{i=1}^\infty \alpha_i \kappa(\cdot, \x_i),\; \alpha_i \in \CC,\; \x_i \in \mathcal{X} \right\}.
\end{equation}
In order to be a Hilbert space, an inner product needs to be defined. It can be constructed by a standard procedure as follows. First, one defines a set of functions that have easy-to-use (e.g. finite dimensional) and desired properties on which such an inner product can exist. Then one argues that this set is dense in $\mathcal{H}_R$ and thus the properties in the limit are true as well. To be more precise, given the set of functions
\begin{equation}
    \mathcal{F}' = \left\{ f : f(\cdot) = \sum_{i=1}^m \alpha_i \kappa(\cdot, \x_i),\; \alpha_i \in \CC,\; \x_i \in \mathcal{X},\; m \in \NN \right\},
\end{equation}
a sesquilinear form is defined as
\begin{equation}
\label{eq:RKHS_inner_product}
    \varphi(f, g) = \sum_{i=1}^m\sum_{j=1}^{m'} \overline{\alpha}_i \beta_j \kappa(\x_i, \x_j') 
\end{equation}
for $f,g\in\mathcal{F}'$. Let us recall what a sesquilinear form actally is. 
\begin{definition}\label{def:sesquilinear_form}
Let $\varphi: \mathcal{H} \times \mathcal{H} \rightarrow \CC$ on a Hilbert space $\mathcal{H}$ with
\begin{align}
    \varphi(x+y,w+v) &= \varphi(x,w) + \varphi(x,v) + \varphi(y,w) + \varphi(y,v),\nonumber \\
    \varphi(ax, bw) &= \overline{a}b\varphi(x,w)\nonumber
\end{align}
for $a,b\in\CC$ and $x,y,w,v \in \mathcal{H}$. Then $\varphi$ is called a \textbf{sesquilinear form}. Furthermore, if $\varphi(x,w) = \overline{\varphi(w,x)}$, then it is called a \textbf{Hermitian sesquilinear form}.
\end{definition}
\noindent Then $\varphi$ can be shown to be independent of the representation (the coefficients) and has all properties that an inner product holds. The feature map can be defined as $\Phi(\x) = \kappa(\cdot, \x) \in \CC^\mathcal{X}$ and $\varphi(\kappa(\cdot, \x), \kappa(\cdot, \x')) = \kappa(\x, \x')$. As the set of functions $\mathcal{F}'$ is dense in $\mathcal{H}_R$ with respect to the norm induced by Eq.~(\ref{eq:RKHS_inner_product}) in $ \mathcal{X}$, one can show that the representation in the limit of $m\rightarrow \infty$ also has all described properties. Therefore the space $\mathcal{H}_R$ is a Hilbert space and we denote the inner product as $\langle f, g \rangle_R = \varphi(f, g)$ for $f,g \in \mathcal{H}_R$ to show that the sesquilinear form is an inner product.

As a consequence one can use the \textit{Representer Theorem}~\cite{scholkopf2001generalized} which roughly states that a function $f \in \mathcal{H}_R$ minimizing the cost $\mathcal{L}(f) = c\left(\left(\x_1, y_1, f\left(\x_1\right)\right), \ldots, \left(\x_M, y_M, f\left(\x_M\right)\right)\right) + \gamma\left(\|f\|\right)$ with strictly monotonically increasing $
\gamma: \RR_{\geq 0} \rightarrow \RR \cup \{\infty\}$ can be represented by a finite sum of kernel evaluations
\begin{equation}
    f(\cdot) = \sum_{i=1}^M \alpha_i \kappa(\cdot, \x_i).
\end{equation}
For the SVM, this means that the calculated optimal separating hyperplane can be represented by a finite evaluation of weighted kernel evaluations. This helpful result has been adapted to quantum machine learning by Ref.~\cite{PhysRevLett.122.040504} and opened up a solid mathematical framework for quantum SVM algorithms. Moreover, as shown in Sec.~\ref{sec:review}, given trained weights (i.e., Karush-Kuhn-Tucker multipliers) the swap-test classifier represents the function that optimally separates the training data from each other. The discovery of this connection is imperative as it gives a thorough mathematical understanding of quantum kernel-based classifiers.

\subsection{Positive Definiteness}
Here we address the central question, whether the similarity measures $\kappa: \mathcal{X} \times \mathcal{X} \rightarrow \FF$ ($\FF = \CC$ or $\RR$) used in the STC are indeed positive definite (and symmetric when $\FF = \RR$) kernels. Given a Hermitian sesquilinear form we find
\begin{equation}
    \sum_{i.j} \overline{c}_i c_j \varphi(\x_i, \x_j)
    = \sum_{i.j} \varphi(c_i \x_i, c_j\x_j) 
    = \varphi\left( \sum_i c_i \x_i, \sum_j c_j \x_j \right) \in \RR \label{eq:sesquilinear_psd}
\end{equation}
for any $\x_i,\x_j \in \mathcal{H}$. As a consequence, it is only necessary to show that $\varphi(\x, \x) \geq 0$ for all $\x \in \mathcal{H}$ to deduce that $\varphi$ is PSD. If $\varphi$ is an inner product, we know that it is a Hermitian sesquilinear form but with $\langle \x, \x \rangle \geq 0$. Therefore, in general, any inner product is PSD. 

\subsubsection{Hilbert-Schmidt Inner Product}
As we have seen in the most general setting of Eq.~(\ref{eq:mixed}), we have the Hilbert-Schmidt inner product $\langle A,B\rangle_{HS} = \Tr(A^\dagger B)$ as the Hermitian sesquilinear form. Let $A$ be a bounded linear operator on $\mathcal{H}$. Then $A^\dagger A$ is Hermitian and non-negative and as a consequence $\Tr(A^\dagger A) \geq 0$. We can conclude that together with Eq.~(\ref{eq:sesquilinear_psd}) the Hilbert-Schmidt inner product is PSD and hence a kernel. Obviously, the restriction from the original space of Hilbert-Schmidt operators to density matrices also retains the PSD property: if $A = \rho_1$ and $B = \rho_2$ are two density matrices, then $\Tr(\rho_1^\dagger \rho_2) = \Tr(\rho_1 \rho_2)$ is a kernel. A further restriction to pure states, i.e., $\rho_1 = \ketbra{\x}{\x}$ and $\rho_2 = \ketbra{\y}{\y}$, must also retain this property. In this case, $\Tr(\rho_1\rho_2) = \Tr(\ketbra{\x}{\x}\cdot\ketbra{\y}{\y}) = |\braket{\x}{\y}|^2$. Thus the squared overlap function is also a valid kernel.

\subsubsection{Alternative elementary proofs}

As an alternative to the above reasoning, we provide elementary proofs for the squared overlap function and the trace of product of two density matrices.
\begin{proposition}\label{thm:fidelity_is_kernel}
Given a quantum Hilbert space $\mathcal{H}$, the function $\kappa: \mathcal{H} \times \mathcal{H} \rightarrow \mathbbm{R}$ as defined 
\begin{equation}
\label{eq:statefidelity}
    \kappa(\ket{x},\ket{y}) = \left| \braket{x}{y} \right|^2
\end{equation}
is symmetric and positive definite, i.e., $\sum_{i,j} a_i a_j \kappa(x_i, x_j) \geq 0$ for all finite sets of $S = \{x_1, \ldots, x_M\} \subset \mathcal{X}$ and all $a_i \in \mathbbm{R}$.
\end{proposition}
\begin{proof}
The proof can be looked up in~\cite{shawetaylor2004kernel} and is done by applying the theory of Hermitian sesquilinear forms (Def.~\ref{def:sesquilinear_form}), rather than symmetric bilinear forms as is usually applied in the domain of machine learning. 
We define two functions $k_1(\ket{x}, \ket{y}) = \braket{x}{y}$ and $k_2(\ket{x}, \ket{y}) = \overline{\braket{x}{y}}$ and first show that both of them are kernels. Then we will show that the product of two kernels is also a kernel. As seen by the properties of the inner product, it can be shown that $k_1$ is PSD, cf. Def.~\ref{def:psd}
\begin{equation}
    \sum_{i,j} \overline{c}_i c_j \braket{\x_i}{\x_j}
    = \sum_{i,j} (c_i\ket{\x_i})^\dagger(c_j\ket{\x_j})
    = \underbrace{\Bigg{(}\sum_i c_i\ket{\x_i}\Bigg{)}^\dagger}_{=:\bra{\phi}} \underbrace{\Bigg{(}\sum_j c_j\ket{\x_j}\Bigg{)}}_{=:\ket{\phi}}
    = \braket{\phi}{\phi}
    = \left\| \phi \right\|^2 \geq 0.
\end{equation}
These steps are valid by linearity in the first argument and conjugate symmetry. Therefore the inner product is by natural choice positive semi-definite.
Now for the conjugate case:
\begin{equation}
    \sum_{i,j} \overline{c}_i c_j \overline{\braket{\x_i}{\x_j}}
    = \sum_{i,j} \left( \overline{c}_j \ket{\x_j} \right)^\dagger \left( \overline{c}_i \ket{\x_i} \right)
    = \underbrace{\Bigg{(}\sum_i \overline{c}_i\ket{\x_i}\Bigg{)}^\dagger}_{=:\bra{\phi}} \underbrace{\Bigg{(}\sum_j \overline{c}_j\ket{\x_j}\Bigg{)}}_{=:\ket{\phi}}
    = \braket{\phi}{\phi}
    = \left\| \phi \right\|^2 \geq 0.
\end{equation}
Since both $k_1$ and $k_2$ are PSD, we can continue to look at the product $k_1 k_2$. Let $K_1$ and $K_2$ be the Gram matrices of $k_1$ and $k_2$, respectively. The Kronecker product of the Gram matrices $K=K_1 \otimes K_2$ is positive semi-definitive if $K_1$ and $K_2$ are~\cite{shawetaylor2004kernel}, which was just proved. The proof is finalized by noting that the Schur (or Hadamard) product $H = K_1 \circ K_2$ is a sub-matrix of $K$, i.e., for any $x \in X$ we can find $\tilde{x} \in X \otimes X$ such that $x^\dagger H x = \tilde{x}^\dagger K \tilde{x} \geq 0$. The only remaining fact to be checked is that $\kappa$ is symmetric which is obvious.
\end{proof}

Now we show that the trace of the product of two density matrices is also positive semi-definite. This is a consequence of the definition of the trace of two density matrices. The space of the linear operators $\rho_m$ for $m=1, \ldots, M$ defined in Eq.~(\ref{eq:data_mixed}) is denoted as $\mathcal{D} \subset \mathcal{B}(\mathcal{H})$. Then we find the following lemma as a consequence of Proposition~\ref{thm:fidelity_is_kernel}.
\begin{lemma}
The function $\kappa_\text{tr}: \mathcal{D} \times \mathcal{D} \rightarrow \mathbbm{R}$ with $ \kappa_\text{tr}(\rho_1, \rho_2) = \Tr(\rho_1\rho_2)$ is positive definite.
\end{lemma}
\begin{proof}
We note that $\Tr(\rho_n \rho_m) = \sum_{i,j} p_{i,n} p_{j,m} \left| \braket{\x_{i,n}}{\x_{j,m}} \right|^2 = \sum_{i,j} p_{i,n} p_{j,m} \kappa(\x_{i,n}, \x_{j,m})$. By an elementary calculation, we see that for $a_n,a_m\in\RR$
\begin{align}
    \sum_{n,m} a_n a_m \Tr(\rho_n\rho_m)
    &= \sum_{n,m} a_n a_m \sum_{i=1}^{M_n} \sum_{j=1}^{M_m} p_{i,n} p_{j,m} \kappa(\x_{i,n}, \x_{j,m})\nonumber \\
    &= \sum_{n,m} a_n a_m \sum_{i=1}^{M_n} \sum_{j=1}^{M_m} p_{i,n} p_{j,m} \braket{\Phi(\x_{i,n})}{\Phi(\x_{j,m})}\nonumber \\
    &= \underbrace{\left( \sum_n\sum_i a_n p_{i,n} \ket{\Phi(\x_{i,n})} \right)^\dagger}_{=: \bra{\phi}} \underbrace{\left( \sum_m\sum_j a_m p_{j,m} \ket{\Phi(\x_{j,m})} \right)}_{=: \ket{\phi}}\nonumber \\
    &= \| \phi \|^2 \geq 0
\end{align}
The basic argument relies on the fact that $\kappa$ is a kernel and thus there exists a feature map $\Phi$ such that it is an inner product of an adequately constructed Hilbert space. Also it is important to note that we use the real-valued version of Def.~\ref{def:psd} as density matrices are Hermitian and non-negative with trace $1$, therefore the trace of the product is also real and non-negative. The last fact to be verified is that $\kappa_{tr}$ is symmetric by the cyclic property of the trace.
\end{proof}

With the help of Mercer's theorem, an interesting observation on the STC is made. There exists a Hilbert space (the RKHS) with a feature map $\Phi: \mathcal{X} \rightarrow \mathcal{H}_{R}$, such that the STC can be represented as
\begin{align}
    f(\testx) = \sum_m (-1)^{y_m} a_m \braket{\Phi(\x_m)}{\Phi(\testx)} = \braket{\Phi_0}{\Phi(\x)} - \braket{\Phi_1}{\Phi(\x)}
\end{align}
with $\Phi_l = \sum_{m|y_m = l} a_m \Phi(\x_m)$ being the feature space class-centroids in the RKHS. This is an estimation of the Bayes decision rule, see e.g. Ref.~\cite{hofmann2008_kernel_methods_ml}. It in essence means that if there are two values $\hat{\x}_0, \hat{\x}_1 \in \mathcal{X}$ such that $\Phi(\hat{\x}_l) = \Phi_l$, then the STC is equivalent to applying the classification to two training data.
\section{Conclusion}
\label{sec:conclusion}
In this work, we reviewed and extended the general theory of quantum kernel-based classifiers. We compared among the quantum support vector machine (qSVM), the Hadamard classifier (HC) and the swap-test classifier (STC). While doing so we showed that either the HC or the STC can be used interchangeably as classifier of the qSVM algorithm. For the extension of the existing theory, we focused on the STC since it employs the similarity measure that is more natural to quantum states than that used in the HC or the qSVM. We showed that the number of qubits necessary for implementing the STC can be further reduced compared to its original construction in Ref.~\cite{2019arXiv190902611B}. The STC was originally shown to be based on the quantum state fidelity (squared overlap) between pure quantum states. Here we generalized the STC for data encoded in density matrices and showed that in fact the STC is based on the Hilbert-Schmidt inner product between two density matrices. Furthermore, by modifying the initial state, the STC can also perform ensemble learning. We showed that the expectation measurement used in the original work of the STC can be replaced with projective measurement and calculated the probability of misclassification. Finally, we explicitly showed that the Hilbert-Schmidt inner product, and hence the squared overlap between two pure states, are symmetric and positive semi-definite kernels. Thus the STC can inherit the mathematical framework of the kernel method developed in classical machine learning.

Interesting future work includes the extension of the quantum-kernel based classifier to multi-label classification, supervised learning with noisy labels, and unsupervised learning algorithms. The development of quantum machine learning algorithms based on other distance measures such as the Bures distance, which can be defined using the quantum state fidelity function for general density matrices, and the trace distance is also a prominent future research topic.

\section*{Acknowledgements}
This research is supported by the National Research Foundation of Korea (Grant No. 2019R1I1A1A01050161 and 2018K1A3A1A09078001) and the Ministry of Science and ICT, Korea, under an ITRC Program, IITP-2019-2018-0-01402, and by the South African Research Chair Initiative of the Department of Science and Technology and the National Research Foundation.

\bibliography{references}
\bibliographystyle{unsrt} 
\end{document}